\documentclass[11pt]{article}
\usepackage[utf8]{inputenc}
\usepackage{hyperref}
\usepackage{amsmath,amsthm,amssymb,caption,subcaption,cleveref,nicefrac}
\usepackage{fullpage}
\usepackage{bm}
\usepackage{algorithm}
\usepackage[noend]{algorithmic}
\usepackage[table]{xcolor}
\usepackage{multirow}
\usepackage{subcaption}
\usepackage{arydshln}
\usepackage{csquotes}

\newtheorem{theorem}{Theorem}

\newtheorem{lemma}[theorem]{Lemma}
\newtheorem{corollary}[theorem]{Corollary}

\newtheorem{definition}[theorem]{Definition}

\newcommand{\cX}{\mathcal{X}}
\newcommand{\cM}{\mathcal{M}}

\newcommand{\ba}{\mathbf{a}}

\newcommand{\eps}{\epsilon}
\newcommand{\N}{\mathbb{N}}

\newcommand{\R}{\mathbb{R}}

\newcommand{\bv}{\mathbf{v}}

\newcommand{\bw}{\mathbf{w}}

\newcommand{\bz}{\mathbf{z}}

\newcommand{\bx}{\mathbf{x}}
\newcommand{\bA}{\mathbf{A}}

\newcommand{\cP}{\mathcal{P}}
\newcommand{\cL}{\mathcal{L}}
\renewcommand{\phi}{\varphi}

\newcommand{\myin}{\mathrm{in}}
\newcommand{\myout}{\mathrm{out}}
\newcommand{\mystart}{\mathrm{start}}
\newcommand{\myend}{\mathrm{end}}

\newcommand{\hN}{\hat{N}}
\newcommand{\hu}{\hat{u}}
\newcommand{\hc}{\hat{c}}

\newcommand{\tf}{\tilde{f}}
\newcommand{\tG}{\tilde{G}}
\newcommand{\tbw}{\tilde{\bw}}
\newcommand{\tw}{\tilde{w}}

\newcommand{\cC}{\mathcal{C}}

\DeclareMathOperator{\Lap}{Lap}

\DeclareMathOperator{\rnn}{\mbox{$r$}-nnd}
\DeclareMathOperator{\argmin}{argmin}

\DeclareMathOperator{\disc}{disc}
\DeclareMathOperator{\dist}{dist}
\DeclareMathOperator{\tdist}{\widetilde{dist}}
\DeclareMathOperator{\hdist}{\widehat{dist}}

\usepackage[textsize=footnotesize,color=green!40]{todonotes}

\allowdisplaybreaks

\title{Differentially Private All-Pairs Shortest Path Distances:\\ Improved Algorithms and Lower Bounds}
\author{
Badih Ghazi \\
Google Research \\
{\small \texttt{badihghazi@gmail.com}}
\and
Ravi Kumar \\
Google Research \\
{\small \texttt{ravi.k53@gmail.com}}
\and
Pasin Manurangsi \\
Google Research \\
{\small \texttt{pasin@google.com}}
\and
Jelani Nelson \\
UC Berkeley \& Google Research \\
{\small \texttt{minilek@alum.mit.edu}}}

\date{\today}

\begin{document}

\maketitle

\begin{abstract}
We study the problem of releasing the weights of all-pair shortest paths in a weighted undirected graph with differential privacy (DP). In this setting, the underlying graph is fixed and two graphs are neighbors if their edge weights differ by at most $1$ in the $\ell_1$-distance. We give an $\eps$-DP algorithm with additive error $\tilde{O}(n^{2/3} / \eps)$ and an $(\eps, \delta)$-DP algorithm with additive error $\tilde{O}(\sqrt{n} / \eps)$ where $n$ denotes the number of vertices. This positively answers a question of Sealfon~\cite{Sealfon16,DPorg-open-problem-all-pairs}, who asked whether a $o(n)$-error algorithm exists. We also show that an additive error of $\Omega(n^{1/6})$ is necessary for any sufficiently small $\eps, \delta > 0$. 

Finally, we consider a relaxed setting where a multiplicative approximation is allowed. We show that, with a multiplicative approximation factor $k$,
the additive error can be reduced to $\tilde{O}\left(n^{1/2 + O(1/k)} / \eps\right)$ in the $\eps$-DP case and $\tilde{O}(n^{1/3 + O(1/k)} / \eps)$ in the $(\eps, \delta)$-DP case, respectively.
\end{abstract}

\section{Introduction}

In recent years, differential privacy (DP)\cite{DworkMNS06, dwork2006our} has emerged as a popular notion for quantifying the leakage of sensitive user information by algorithms for a variety of tasks, from data analytics to machine learning, and seen increasing adoption in industry \cite{dp2017learning, ding2017collecting, google_dp} and government agencies \cite{abowd2018us}. Loosely speaking, DP requires that the output of a (randomized) algorithm remains statistically near-indistinguishable when the input dataset is replaced by another so-called “neighboring” dataset, which only differs on the contributions of a single user. DP has been applied to graph problems, including the all-pairs shortest distances problem defined next.

In the \emph{all-pairs shortest path distances (APSD)} problem, we are given a weighted undirected graph $G = (V, E, \bw)$. The goal is to output, for all $u, v \in V$, an estimate $\tdist(u, v)$ of the smallest-weight distance $\dist_G(u, v)$ among all paths from $u$ to $v$.  APSD is among the most well-studied algorithmic graph problems, with known classic algorithms such as Floyd-–Warshall that can solve the problem in $O(n^3)$ time. More recently, (slightly) faster algorithms have been discovered~\cite{Williams18}. Additionally, the hypothesis that the problem cannot be solved in $O(n^{3 - \Omega(1)})$ time serves as one of the main starting points in fine-grained complexity theory (see e.g.,~\cite{WilliamsW18}).

Sealfon~\cite{Sealfon16} was the first to formally study APSD with privacy.  In the DP framework, 
we require that our algorithm be $(\eps, \delta)$-DP, where---following~\cite{Sealfon16}\footnote{Note that, similar to~\cite{Sealfon16}, we have to keep the underlying graph $(V, E)$ fixed and we are only allowed to change the weights. This is required to get any non-trivial result: as pointed out in~\cite{Sealfon16}, removing an edge can make the graph disconnected and therefore the distance can go from finite to infinite.}---the neighboring graphs correspond to those whose weight vector $\bw$ differs by at most 1 in the $\ell_1$-distance. Throughout, we measure the \emph{error} in terms of the maximum (i.e., $\ell_\infty$) error, defined as $\max_{u, v \in V} |\dist_G(u, v) - \tdist(u, v)|$. We say that an output is \emph{$\alpha$-accurate} iff the error is at most $\alpha$. We use $n$ to denote the number of vertices.  Sealfon gave an $O(n \log n / \eps)$-error ``input-perturbation'' algorithm for APSD, which adds Laplace noise to all edge weights and computes the shortest path in this resulting graph with noisy weights. Furthermore, he showed that, if we are required to release the shortest \emph{paths} themselves, then this error is tight up to a logarithmic factor. However, this lower bound does not apply to the APSD problem, where only the weights have to be released.

Given the above state of the problem, Sealfon~\cite{Sealfon16,DPorg-open-problem-all-pairs} asked the following question:
\begin{displayquote}
{\sl
Is it possible to privately release all-pairs distances with error sublinear in $n$?
}
\end{displayquote}

\subsection{Our Results}

We positively answer this question, by providing private algorithms with sublinear error:
\begin{theorem} \label{thm:pure-intro}
For any $\eps > 0$, there is an $\eps$-DP algorithm that is\footnote{
For readability, our upper bounds in this section suppress factors that are polylogarithmic in $n$ and $1/\delta$.  We use $\tilde{O}$ to hide a dependency of the form $(\log n)^{O(1)}$ in the $\eps$-DP case and $(\log n \cdot \log(1/\delta))^{O(1)}$ in the $(\eps, \delta)$-DP case.} $\tilde{O}(n^{2/3}/\eps)$-accurate w.h.p. 
\end{theorem}
\begin{theorem} \label{thm:apx-intro}
For any $\eps, \delta \in (0, 1)$, there is an $(\eps, \delta)$-DP algorithm that is $\tilde{O}(\sqrt{n}/\eps)$-accurate w.h.p. 
\end{theorem}

We also prove a lower bound of $\Omega(n^{1/6})$ in terms of accuracy, which not only affirms that polynomial dependency in $n$ is necessary but---in light of $O(\log^{2.5} n)$-accurate algorithm on trees~\cite{Sealfon16}---also separates APSD on general graphs from that on trees.


\begin{theorem} \label{thm:lb-main}
For any $\beta \in (0, 1)$ and any sufficiently small $\eps, \delta > 0$ (depending only on $\beta$), no $(\eps, \delta)$-DP algorithm for APSD can be $o(n^{1/6})$-accurate with probability $1 - \beta$.
\end{theorem}

Finally, we consider a slightly more relaxed setting where a \emph{multiplicative} approximation is also allowed. More precisely, we say that an output estimate $\tdist$ is \emph{$(\gamma, \alpha)$-accurate} iff, for all $u, v \in V$, we have
\begin{align*}
\dist_G(u, v) - \alpha \leq \tdist(u, v) \leq \gamma \cdot \dist_G(u, v) + \alpha.
\end{align*}
Multiplicative approximation has long been widely studied in the non-private setting (e.g.,~\cite{AwerbuchBCP98,Cohen98,ThorupZ05}). Under this relaxed setting, we give algorithms with improved additive error guarantees. Specifically, for a multiplicative error $\gamma > 1$, the additive error for our $\eps$-DP and $(\eps, \delta)$-DP algorithms are $\tilde{O}(n^{1/2 + O(1/\gamma)} / \eps)$ and $\tilde{O}(n^{1/3 + O(1/\gamma)} / \eps)$ respectively.

\begin{theorem} \label{thm:pure-mult-intro}
For any $\eps > 0$ and $k \in \N$, there is an $\eps$-DP algorithm that is $\left(2k - 1, \tilde{O}_k(n^{(k+1)/(2k+1)}/\eps)\right)$-accurate w.h.p. 
\end{theorem}

\begin{theorem} \label{thm:apx-mult-intro}
For any $\eps, \delta \in (0, 1)$ and $k \in \N$, there is an $(\eps, \delta)$-DP algorithm that is \\ $\left(2k - 1, \tilde{O}_k(n^{(k+1)/(3k+1)}/\eps\right)$-accurate w.h.p. 
\end{theorem}

We remark that our lower bound (\Cref{thm:lb-main}) does \emph{not} apply to the case where multiplicative approximation is allowed, and it remains an interesting open question whether it is possible to achieve an $n^{o(1)}$ additive error in this case.

\subsection{Technical Overview}

\paragraph{Additive-Approximation Algorithms.}
We start with algorithms that only incur an additive error (\Cref{thm:pure-intro,thm:apx-intro}). Our algorithms will combine two primitives: the aforementioned ``input-perturbation'' algorithm and an ``output-perturbation'' algorithm. For the latter (also noted in~\cite{Sealfon16,DPorg-open-problem-all-pairs}), we simply add Laplace noise to each of the shortest path distances. Since each shortest path distance has sensitivity one, it suffices to take the Laplace noise scale to be $n^2/\eps$ to achieve $\eps$-DP. This yields an error of $\tilde{O}(n^2 / \eps)$ for the problem.

To combine the two, we need a notion of $t$-hop path. For $t \in \N$, a path is said to be a \emph{$t$-hop path} if it consists of at most $t$ edges. As proved in~\cite{Sealfon16}, for any pair $u, v$ for which the shortest path is a $t$-hop path, the error of the input-perturbation algorithm is in fact just $\tilde{O}(t/\eps)$ instead of the worst-case bound of $\tilde{O}(n/\eps)$.

Our strategy is now as follows: we sample a subset $S \subseteq V$ of a given size $s$.  Vertices in $S$ will serve as intermediate vertices for long paths. We compute the shortest path distances between every pair of vertices in $S$ using the output-perturbation algorithm, whose error is now only $\tilde{O}(s^2 / \eps)$. We also use the input-perturbation algorithm to compute shortest $t$-hop path distances between all pairs of vertices in $V$; this has an error of $\tilde{O}(t/\eps)$. Finally, for each pair $u, v$, we compute the distance between them by taking the minimum of (i) the shortest $t$-hop path distance between $u$ and $v$ and (ii) the minimum over all $w, z \in S$ of the sum of the shortest $t$-hop path distance between $u, w$, the shortest path distance between $w, z$ and the shortest $t$-hop path distance between $z, v$.

For this algorithm to work, $t$ has to be sufficiently large so that every shortest path contains a vertex in $S$ among its first $t$ and last $t$ points. It is not hard to see that this w.h.p. holds as long as we pick $t = \tilde{O}(n/s)$. In total, the error is then $\tilde{O}(s^2 / \eps) + \tilde{O}(t/\eps) = \tilde{O}\left(\frac{1}{\eps}(s^2 + n/s)\right)$. This is minimized when we take $s = \tilde{\Theta}(n^{1/3})$, which yields an error of $\tilde{O}(n^{2/3}/\eps)$ as claimed in \Cref{thm:pure-intro}.

In the case of $(\eps, \delta)$-DP, we may use the Gaussian mechanism instead of the Laplace mechanism. This reduces the error in the output-perturbation algorithm to $\tilde{O}(s/\eps)$. By taking $s = \tilde{\Theta}(\sqrt{n})$, we get the claimed bound in \Cref{thm:apx-intro}.

We end by noting that in our actual algorithms we do not use the Laplace/Gaussian mechanisms for computing the shortest path distances between points in $S$; rather, we use recent improved results~\cite{SteinkeU16,DaganKur20,GKM21} that allow us to remove a logarithmic bound in the error.

\paragraph{Lower Bounds.}
Our lower bounds follow from a reduction from the \emph{linear queries}\footnote{Since we are using only Boolean matrices $\bA$ here, this version is sometimes called \emph{counting queries}.} problem. A set of linear queries can be encoded as a matrix $\bA \in \{0, 1\}^{d \times n}$. The input to the algorithm is a vector $\bz \in \{0, 1\}^n$ (where two datasets $\bz, \bz'$ are neighbors iff they differ on a single coordinate). The goal is to output an estimate of $\bA\bz$. We write $\ba_i$ as the $i$th row of $\bA$.

To encode linear queries into a graph, we let each $z_i$ be a weight of some edge in our graph, and the goal is to identify each linear query $\ba_i$ with a shortest path. Our goal is that the shortest path has length $\left<\ba_i, z\right>$ (perhaps plus some constant that we know beforehand). Although this seems promising, 
it is not simple to embed linear queries into a graph in this way.
Consider, for example, two linear queries $z_1 + z_2 + z_4$ and $z_1 + z_3 + z_4$. Here, if the first one corresponds to a path that has $z_1, z_2, z_4$ in this order and the second corresponds to a path that has $z_1, z_3, z_4$ in this order, then clearly both of them cannot be the shortest path at the same time.

To avoid the issue described above, we embed the graph while ensuring that each linear query appears in a \emph{unique shortest path} between a pair of endpoints. Roughly speaking, this is done by taking a graph $H = (U, E_H, w_H)$ that has many unique shortest paths, then scaling the edge weights so that the gap between any shortest path and the second shortest path (w.r.t. the same pair) is sufficiently large. Then, we replace each vertex $i \in U$ by an edge, where the weight is equal to $z_i$. 
(Note that, strictly speaking, such a replacement is only well-defined when there is an ordering on the vertices; this is formalized in \Cref{sec:lb}.)
In this way, each linear query includes all the vertices presented in a corresponding unique shortest path.

To obtain a concrete bound, we use a graph $H$ where each vertex is a point in the plane and there is an edge between two vertices iff the line segment between them does not contain another vertex and the weight is simply the Euclidean distance. In this case, the unique shortest paths are simply lines in the plane. These ``point-line'' set systems have been studied before, and we appeal to a discrepancy\footnote{For point-line system, the discrepancy is the minimum across all $2$-colorings of all points of the maximum difference between the number of each color on any line. For a more formal definition, see \Cref{def:discrepancy}} lower bound of $\Omega(n^{1/6})$ of such a set system by Chazelle and Lvov~\cite{ChazelleL01a}; together with the known connection between discrepancy and linear queries lower bounds~\cite{MuthukrishnanN12}, we arrive at our lower bound for the APSD problem. 

On a tangent, the original work of Chazelle and Lvov~\cite{ChazelleL01a} proved a nearly tight upper bound of $O(n^{1/6} \log^{2/3} n)$ for point-line systems. We observe that a sharper partial coloring lemma shown by Lovett and Meka~\cite{LovettM15} immediately implies an improved---and tight---upper bound of $O(n^{1/6})$, which might be of independent interest. The details of this proof can be found in \Cref{sec:discrepancy-upper-bound}.

\paragraph{Additive and Multiplicative Approximation Algorithms.}
When also allowed multiplicative approximation, we use the algorithm of Thorup and Zwick~\cite{ThorupZ05}, which constructs a $(2k - 1)$-stretch spanner with only $O_k(n^{1 + 1/k})$ edges. We privatize their algorithm by replacing each non-private primitive with a private one (e.g., selecting a closeby point is now done using the Exponential Mechanism~\cite{McSherryT07}). Roughly speaking, the benefit of this algorithm is that, unlike in the ``naive'' output-perturbation algorithm we described above that queries the distances of all $O(n^2)$ pairs of vertices, here we only need to query the shortest path distances of only $O_k(n^{1 + 1/k})$ pairs of vertices. Therefore, this also reduces the noise needed for each query to attain the same level of privacy. By using this algorithm to compute the shortest paths within $S$, we can thus reduce the additive error of the overall algorithm, albeit at the price of incurring a multiplication error.

\subsection{Related Work}

As alluded to earlier, our work is most closely related to that of Sealfon~\cite{Sealfon16}. In addition to the already discussed results,~\cite{Sealfon16} also considered several special cases, including an $\eps$-DP $O(\log^{2.5} n / \eps)$-accurate algorithm for trees. It is worth noting that this is a generalization of the prefix sum (aka one-dimensional range query) problem and Sealfon's algorithm can be considered a generalization of previous algorithms for prefix sum~\cite{DworkNPR10,ChanSS11,DworkNRR15}. Another result from~\cite{Sealfon16} concerns the case where each weight is bounded by $M$ (i.e., $w_e \leq M$ for all $e \in E$), in which case an algorithm with error $\tilde{O}(\sqrt{n M})$ was shown. Interestingly, Sealfon's algorithm also employs the idea of $t$-hop paths. However, instead of trying to find $S$ that covers the \emph{shortest paths} (as we do), his algorithm directly finds $S$ that covers the \emph{original graph} $G$ via $t$-hop paths. Note that this can be computed with no privacy loss at all. The algorithm then proceeds to compute the all-pair shortest path distances for $S$ via the output-perturbation algorithm. The main bottleneck here is that, since we are only guaranteed that there are \emph{some} $t$-hop paths from any vertex $v \in V$ to $S$ but this path may not coincide with (a subpath of) the shortest path, the error bound becomes $t \cdot M$ (because we can only guarantee a bound of $M$ on each edge); this causes a dependency on $M$ in the final bound.

\section{Background}

For any $m \in \N$, we write $[m]$ as a shorthand for $\{1, \dots, m\}$.  Let $x \sim D$ denote that the random variable $x$ is sampled according to the distribution $D$.  We use the following notion of neighboring datasets throughout the paper: two input datasets $X, X'$ are said to be \emph{neighbors} if $\| X - X' \|_1 \leq 1$.

We first recall the definition of differential privacy (DP). 
\begin{definition}[Differential Privacy]
An algorithm $\cM$ is \emph{$(\eps, \delta)$-differentially private ($(\eps, \delta)$-DP)} iff for every neighboring input datasets $X, X'$ and every set $S$ of possible outcomes, we have $\Pr[\cM(X) \in S] \leq e^\eps \cdot \Pr[\cM(X') \in S] + \delta$. When $\delta = 0$, we simply write $\eps$-DP instead of $(\eps, \delta)$-DP.
\end{definition}

\begin{definition}[Sensitivity]
The \emph{$\ell_p$ sensitivity}  of $f: \cX \to \R^D$ is defined as $\Delta_p(f) := \max_{X, X'} \|f(X) - f(X')\|_p$, where the maximum is over all neighboring datasets $X, X'$.
\end{definition}

We will also use the Laplace mechanism~\cite{DworkMNS06}, which we recall here. Let $\Lap(b)$ denote the Laplace distribution%
\footnote{The pdf of the Laplace distribution with scale $b$ is $\propto \exp(-|x|/b)$.} with scale $b$. 
\begin{lemma}[Laplace Mechanism~\cite{DworkMNS06}]
\label{lem:laplace}
Given any function $f: \cX \to \R^D$, the Laplace mechanism, on input $X \in \cX$, outputs 
$f(X) + (\eta_1, \ldots, \eta_D)$ where each
$\eta_i \sim \Lap(\Delta_1(f)/\epsilon)$ is independent. The Laplace mechanism is $\epsilon$-DP, and it is $O(\Delta_1(f) \cdot \log(1/\beta) / \eps)$-accurate with probability $1 - \beta$ for all $\beta \in (0, 0.5)$.
\end{lemma}

Next we state the composition theorems for DP, which will be used in our analysis.

\begin{theorem}[Basic Composition~\cite{DworkMNS06,DworkL09}]
\label{thm:basic-composition}
Let $\eps_1, \dots, \eps_k, \delta_1, \dots, \delta_k > 0$. If we run $k$ (possibly adaptive) mechanisms where the $i$th mechanism is $(\eps_i, \delta_i)$-DP, then the entire algorithm is $(\eps_1 + \cdots + \eps_k, \delta_1 + \cdots + \delta_k)$-DP.
\end{theorem}

We will also use the advanced composition; the version stated below is a special case of~\cite[Corollary 3.21]{DworkR14}.

\begin{theorem}[Advanced Composition~\cite{DworkRV10, DworkR14}] \label{thm:advanced-composition}
Let $\eps, \delta \in (0, 1]$ and $k \in \N$. Suppose that we run $k$ (possibly adaptive) mechanisms where each mechanism is $\left(\frac{\eps}{2\sqrt{2k\ln(2/\delta)}}, \frac{\delta}{2k}\right)$-DP. Then the entire algorithm is $(\eps, \delta)$-DP.
\end{theorem}

\subsection{Answering Multiple Queries}

We use the following results that provide stronger guarantees than the union bound when answering multiple queries. As stated earlier, these help save a logarithmic factor compared to using the Laplace or the Gaussian mechanisms directly.
\begin{theorem}[\cite{SteinkeU16}] \label{thm:pure-DP-avoid-union}
Let $f: \cX \to \R^D$ be any function, $\eps > 0$ and $\beta \in (0, 0.5)$. There exists an $\eps$-DP algorithm $\cM$ such that, for any input $X$, with probability $1 - \beta$, we have $\|\cM(X) - f(X)\|_{\infty} \leq O\left(D + \log(1 / \beta)\right) \cdot \left(\Delta_{\infty}(f) / \eps\right)$.
\end{theorem}

\begin{theorem}[\cite{GKM21}] \label{thm:apx-DP-avoid-union}
Let $f: \cX \to \R^D$ be any function and $\eps, \delta > 0$. There exists an $(\eps, \delta)$-DP algorithm $\cM$ such that, for any input $X$, with probability $1 - O(1/D^{10})$, we have $\|\cM(X) - f(X)\|_{\infty} \leq O\left(\sqrt{D \log(1/\delta)}\right) \cdot \left(\Delta_{\infty}(f) / \eps\right)$.
\end{theorem}

\subsection{Selection via the Exponential Mechanism}

In the \emph{selection} problem, we are given a set $\cC$ of candidates together with sensitivity-$1$ functions $f_c$ for all $c \in \cC$. The goal is to output a candidate $c \in \cC$ with the smallest $f_c(X)$ where $X$ denotes the (private) input. We say that the output $(\hc, \tf_{\hc})$ is \emph{$\alpha$-accurate} iff $f_{\hc}(X) \leq \min_{c \in \cC} f_c(X) + \alpha$ and\footnote{The original exponential mechanism from~\cite{McSherryT07} does not output the estimate $\tf_{\hc}$ but this can easily be done using, e.g., the Laplace mechanism.} $|\tf_{\hc} - f_{\hc}(X)| \leq \alpha$.

The exponential mechanism~\cite{McSherryT07} can solve this problem with error $O(\log |\cC| / \eps)$, stated more precisely below.

\begin{theorem}[\cite{McSherryT07}] \label{thm:em}
For any $\eps > 0$, there is an $\eps$-DP algorithm for selection that with probability $1 - \beta$ is $O\left(\frac{\log(|\cC|/\beta)}{\eps}\right)$-accurate.
\end{theorem}

\subsection{Linear Queries}

A set of \emph{linear queries} is parametrized by a matrix $\bA \in \R^{D \times N}$. The input to the problem is a vector $\bv \in \{0, 1\}^n$ and the goal is to compute $\bA \bv$. As before, two inputs $\bv, \bv'$ are neighbors iff $\|\bv - \bv\|_1 \le 1$.

We exploit the connection between linear queries and discrepancy proved in~\cite{MuthukrishnanN12}. Before stating their result, we need a variant of discrepancy they employ.  
\begin{definition}[Discrepancy] \label{def:discrepancy}
For $\gamma \in (0, 1]$, let $\disc_\gamma(\bA) := \min_{\bz \in \{-1, 0, +1\}^N, \|\bz\|_1 \geq \gamma N} \|\bA\bz\|_\infty$. When $\gamma = 1$, we may write $\disc(\bA)$ as a shorthand for $\disc_\gamma(\bA)$. %
\footnote{
The case $\gamma = 1$ is the standard definition of discrepancy; see e.g.,~\cite{Chazelle-book} for more details on discrepancy theory.
}
\end{definition}

We can now state their results, which allow us to translate lower bounds on the discrepancy of $\bA$ to that of the corresponding linear query problem.

\begin{theorem}[\cite{MuthukrishnanN12}] \label{thm:linearq-to-disc}
For any $\beta \in (0, 1)$, there exist $\gamma, \eps, \delta > 0$ such that, for any $\bA$, no $(\eps, \delta)$-DP mechanism is $\left(\disc_\gamma(\bA)/2\right)$-accurate with probability $1 - \beta$ for the $\bA$-linear query problem.
\end{theorem}

\section{Algorithm}

Recall from the Introduction that a path is said to be a \emph{$t$-hop path} if it consists of at most $t$ edges. We write $\dist_{\bw}^{\leq t}(u, v)$ to denote the smallest-weight among all $t$-hop paths from $u$ to $v$.

\subsection{Input-Perturbation Algorithm}

As described earlier in the Introduction, the ``input-perturbation'' algorithm that adds Laplace noise already produces accurate distances for shortest $t$-hop paths when $t$ is small. This is stated below. Note that~\cite{Sealfon16} already proved a similar statement but with one-sided error; we reprove it here for completeness.

\begin{theorem} \label{thm:input-pert}
Let $t \in \N$ be any positive integer and $\beta \in (0, 0.5)$. There is an $\eps$-DP algorithm for computing $\tdist^{\leq t}(u, v)$ such that with probability $1 - \beta$ we have
\begin{align*}
\max_{u, v \in V} |\tdist^{\leq t}(u, v) - \dist^{\leq t}_G(u, v)| \leq O\left(t \cdot \log(n/\beta) / \eps\right).
\end{align*}
\end{theorem}

\begin{proof}
The algorithm works as follows. First, construct the weight $\tbw$ by applying the $\eps$-DP Laplace mechanism (\Cref{lem:laplace}) to $\bw$, i.e., let $\tw_{(u, v)} = w_{(u, v)} + z_{(u, v)}$ where $z_{(u, v)} \sim \Lap(1/\eps)$. Finally, we let $\tdist^{\leq t}(u, v)$ be the weight of the shortest $t$-hop path\footnote{Note that this can be computed in polynomial-time using a simple dynamic programming procedure: let $\tdist^{\leq 0}(u, u) = 0$ for all $u \in V$ and $\tdist^{\leq 0}(u, v) = \infty$ for all $u \ne v \in V$. Then, for $i \in [t], u, v \in V$, let $\tdist^{\leq i}(u, v) = \min_{(u, w) \in E}\{\tw_{(u, v)} + \tdist^{\leq i - 1}(w, v)\}$.} from $u$ to $v$ in $\tG = (V, E, \tbw)$. Since $\tbw$ is a result of an $\eps$-DP Laplace mechanism and our final output $\tdist$ is a post-processing of $\tbw$, we can conclude that our algorithm is $\eps$-DP.

We next argue its utility. From the utility guarantee of the Laplace mechanism (\Cref{lem:laplace}), we can conclude that with probability $1 - \beta$, we have
\begin{align} \label{eq:weight-accuracy}
\|\tbw - \bw\|_{\infty} \leq O\left(\log(|E| / \beta) / \eps\right) \leq O\left(\log(n/\beta) / \eps\right).
\end{align}
Since each $t$-hop path consists of at most $t$ edges, \eqref{eq:weight-accuracy} implies that the weight of any $t$-hop path in $G$ and that in $\tG$ differ by at most $O\left(t \cdot \log(n/\beta) / \eps\right)$. From our definitions of $\tdist^{\leq t}$ and $\dist^{\leq t}$, the claimed bound immediately follows.
\end{proof}

\subsection{Output-Perturbation Algorithm}
\label{sec:output-pert}

For a subset $S \subseteq V$, we define the \emph{$S$-pair shortest path distances problem} to be the problem of outputting, for all $u, v \in S$, an estimate $\tdist(u, v)$ of the smallest-weight $\dist_G(u, v)$ among all paths from $u$ to $v$. (Note that $\dist_G(u, v)$ is still the distance in the original graph $G$, i.e., the path is allowed to leave the subset $S$.)
Of course, the original APSD problem is simply the case $S = V$.

Notice that each $\dist_G(u, v)$ for all $u, v \in S$ has sensitivity at most one. Therefore, we may apply \Cref{thm:pure-DP-avoid-union} and \Cref{thm:apx-DP-avoid-union} directly to the function $f(\bw) = (\dist_G(u, v))_{u, v \in S}$ to achieve the following guarantees for the $S$-pair shortest path distances problem.

\begin{corollary} \label{cor:output-pure-dp}
For any $S \subseteq V$ and $\eps > 0$, there exists an $\eps$-DP algorithm for $S$-pair shortest path distances problem that is $O\left(|S|^2 / \eps\right)$-accurate with probability $1 - O(|S|^{-10})$.
\end{corollary}

\begin{corollary} \label{cor:output-apx-dp}
For any $S \subseteq V$ and $\eps, \delta > 0$, there exists an $(\eps, \delta)$-DP algorithm for $S$-pair shortest path distances problem that is $O\left(|S| \sqrt{\log(1/\delta)} / \eps\right)$-accurate with probability $1 - O(|S|^{-10})$.
\end{corollary}

\subsection{Combining the Two Algorithms: Beating the Linear Error}

In this section, we show how to combine the input-perturbation and output-perturbation algorithms to get our final error bound. We start with the following generic lemma for combining the two algorithms.  (We state this lemma in such a way that includes the multiplicative error, since we will use it in the latter sections as well.)

\begin{lemma} \label{lem:combining}
Suppose that for any set $S$, there exists an $(\eps, \delta)$-DP algorithm that is $(\gamma_{\eps, \delta}(|S|), \alpha_{\eps, \delta}(|S|))$-accurate with probability at least $1 - \beta_{\eps, \delta}(|S|)$. Then, for any parameter $s \in [|V|]$, there exists an $(\eps, \delta)$-DP algorithm that is $\left(\gamma_{\eps/2, \delta}(s), O\left(\alpha_{\eps/2, \delta}(s) + n/s \cdot \log^2 n / \eps\right)\right)$-accurate with probability $1 - \beta_{\eps/2, \delta}(s) - O(1/n)$.
\end{lemma}

\begin{proof}
The algorithm works as follows:
\begin{itemize}
\item Randomly sample a subset $S \subseteq V$ of size $s$.
\item Run the $(\eps/2, \delta)$-DP algorithm to compute estimates $\tdist^S(u, v)$ for all $u, v \in S$.
\item Run the $(\eps/2)$-DP algorithm from \Cref{thm:input-pert} to compute $\tdist^{\leq t}(u, v)$ for all $u, v \in V$ where $t := \lceil 10 \cdot (n/s) \log n \rceil$.
\item Finally, for all $u, v \in V$, let
\begin{align*}
\tdist(u, v) = \min\left\{\tdist^{\leq t}(u, v), \min_{w, z \in S} \tdist^{\leq t}(u, w) + \tdist^S(w, z) + \tdist^{\leq t}(z, v)\right\},
\end{align*}
and output $\tdist$.
\end{itemize}
Basic composition (Theorem~\ref{thm:basic-composition}) immediately implies that our algorithm is $(\eps, \delta)$-DP. As for the utility, the guarantee of the $S$-pair shortest path distances algorithm implies that, with probability $1 - \beta_{\eps/2,\delta}(s)$, we have 
\begin{align} \label{eq:acc-hubs}
\dist_G(w, z) - \alpha_{\eps/2, \delta}(s) \leq \tdist^S(w, z) \leq \gamma_{\eps/2, \delta}(s) \cdot \dist_G(w, z) + \alpha_{\eps/2, \delta}(s) & &\forall w, z \in S.
\end{align}
Furthermore, \Cref{thm:input-pert} ensures that, with probability $1 - 1/n$, we have
\begin{align}  \label{eq:acc-short-paths}
|\tdist^{\leq t}(u, v) - \dist^{\leq t}_G(u, v)| \leq O\left(t \cdot \log n / \eps\right) & &\forall u, v \in V.
\end{align}
Next, consider any $u, v \in V$. Let $u = p_0, p_1, \dots, p_\ell = v$ denote a shortest path from $u$ to $v$ (ties broken arbitrarily). If $\ell > t$, let $i^{u, v, \min}, i^{u, v, \max}$ denote the smallest and largest indices $i$ such that $p_i \in S$. Since $S$ is a u.a.r. $s$-size subset of $V$, a simple probability argument implies that $i^{u, v, \min} \leq t$ and $i^{u, v, \max} \geq \ell - t$ with probability $O(1/n^2)$. 

Henceforth, we will condition on these three events to happen; this occurs with probability at least $1 - \beta_{\eps/2,\delta}(s) - O(1/n)$. Under this assumption, we have
\begin{align*}
\tdist(u, v) &= \min\left\{\tdist^{\leq t}(u, v), \min_{w, z \in S} \tdist^{\leq t}(u, w) + \tdist^S(w, z) + \tdist^{\leq t}(z, v)\right\} \\
&\overset{\eqref{eq:acc-short-paths},\eqref{eq:acc-hubs}}{\geq} \min\left\{\dist^{\leq t}(u, v), \min_{w, z \in S} \dist^{\leq t}_G(u, w) + \dist_G(w, z) + \dist^{\leq t}_G(z, v)\right\} \\
&\qquad - 2 \cdot \alpha_{\eps/2, \delta}(s) - O\left(t \cdot \log n / \eps\right) \\
&\geq \dist_G(u, v) - O\left(\alpha_{\eps/2, \delta}(s) + n/s \cdot \log^2 n / \eps\right).
\end{align*}

As for the upper bound on $\tdist(u, v)$, consider two cases, based on whether the shortest path between $u, v$ has at least $t$ edges. If not, then we have
\begin{align*}
\tdist(u, v) \leq \tdist^{\leq t}(u, v) \overset{\eqref{eq:acc-short-paths}}{\leq} \dist^{\leq t}(u, v) + O(t \cdot \log n / \eps) = \dist(u, v) + O(n/s \cdot \log^2 n / \eps).
\end{align*}
Otherwise, we also have
\begin{align*}
\tdist(u, v) &\leq \min_{w, z \in S} \tdist^{\leq t}(u, w) + \tdist^S(w, z) + \tdist^{\leq t}(z, v) \\
&\leq \tdist^{\leq t}(u, p_{i^{u, v, \min}}) + \tdist^S(p_{i^{u, v, \min}}, p_{i^{u, v, \max}}) + \tdist^{\leq t}(p_{i^{u, v, \max}}, v) \\
&\overset{\eqref{eq:acc-short-paths},\eqref{eq:acc-hubs}}{\leq}  \dist^{\leq t}(u, p_{i^{u, v, \min}}) + \gamma \cdot \dist^S(p_{i^{u, v, \min}}, p_{i^{u, v, \max}}) + \dist^{\leq t}(p_{i^{u, v, \max}}, v) \\&\qquad + 2 \cdot \alpha_{\eps/2, \delta}(s) + O\left(t \cdot \log n / \eps\right) \\
&\leq \gamma \cdot \left(\dist^{\leq t}(u, p_{i^{u, v, \min}}) + \dist^S(p_{i^{u, v, \min}}, p_{i^{u, v, \max}}) + \dist^{\leq t}(p_{i^{u, v, \max}}, v)\right) \\&\qquad + 2 \cdot \alpha_{\eps/2, \delta}(s) + O\left(t \cdot \log n / \eps\right) \\
&= \gamma \cdot \dist_G(u, v) + O\left(\alpha_{\eps/2, \delta}(s) + O(n/s \cdot \log^2 n / \eps)\right).
\end{align*}
Therefore, our algorithm is $\left(\gamma, O\left(\alpha_{\eps/2, \delta}(s) + n/s \cdot \log^2 n / \eps\right)\right)$ with probability $1 - \beta_{\eps/2, \delta}(s) - O(1/n)$ as desired.
\end{proof}

By plugging~\Cref{cor:output-pure-dp} into \Cref{lem:combining} with $s = \sqrt[3]{n \log^2 n}$, we immediately arrive at the following, which is a more detailed version of \Cref{thm:pure-intro}.

\begin{theorem} \label{thm:pure-dp-alg}
For any $\eps > 0$, there exists an $\eps$-DP algorithm for APSD that is $O\left(\frac{n^{2/3} \log^{4/3} n}{\eps}\right)$-accurate w.h.p.
\end{theorem}

Similarly, by plugging~\Cref{cor:output-apx-dp} into \Cref{lem:combining} with $s = \sqrt{n} \log n / \sqrt[4]{\log(1/\delta)}$, we immediately arrive at the following, which is a more detailed version of \Cref{thm:apx-intro}.

\begin{theorem} \label{thm:apx-dp-alg}
For any $\eps, \delta > 0$, there exists an $(\eps, \delta)$-DP algorithm for APSD that is $O\left(\frac{\sqrt{n} \log n \sqrt[4]{\log(1/\delta)}}{\eps}\right)$-accurate w.h.p.
\end{theorem}

\section{Lower Bound}
\label{sec:lb}

In this section, we prove our lower bound (\Cref{thm:lb-main}).

\subsection{Reducing Linear Queries to Shortest Path}

We start by presenting our reduction from linear queries to APSD. We follow the terminology of~\cite{Bodwin19} for the notion of strongly metrizable path systems.  
\begin{definition}[Strongly Metrizable Path System~\cite{Bodwin19}]
A \emph{path system} is a pair $(U, \Pi)$ where $U$ is the set of \emph{vertices} and $\Pi$ is a set of paths. A path system $(U, \Pi)$ is \emph{strongly metrizable} if there is a weighted graph $H = (U, E_H, w_H)$ such that each path $\pi \in \Pi$ is the unique shortest path between its endpoints in $H$.
\end{definition}

Additionally, an \emph{ordered path system} $(U, \Pi, \prec)$ is a path system $(U, \Pi)$ together with a (total) ordering $\prec$ on $U$ such that every path $\pi \in \Pi$ respects such an ordering, i.e., $\pi = (p_1, \dots, p_\ell)$ satisfies $p_i \prec p_j$ for all $1 \leq i < j \leq \ell$. We say that an ordered path system $(U, \Pi, \prec)$ is strongly metrizable iff the underlying path system $(U, \Pi)$ is strongly metrizable.

Finally, we define the \emph{incidence matrix} of the path system $\Pi$ to be the matrix $\bA^{\Pi}$ whose rows are indexed by $\pi \in \Pi$ and columns are indexed by $u \in U$ such that $\bA^{\Pi}_{\pi, u} = 1$ iff $u \in \pi$.

Below we provide our reduction from linear queries to APSD. This reduction works by replacing each vertex in $U$ by an edge with a weight being the input of the linear query problem, as outlined in the Introduction.

\begin{lemma} \label{lem:reduction}
Let $(U, \Pi, \prec)$ be any strongly metrizable ordered path system. If there exists an $(\eps, \delta)$-DP algorithm that is $\alpha$-accurate with probability $1 - \beta$ for APSD on all $2|U|$-vertex graphs, then there exists an $(\eps, \delta)$-DP algorithm that is $\alpha$-accurate with probability $1 - \beta$ for $\bA^{\Pi}$-linear query.
\end{lemma}

\begin{proof}
Since $(U, \Pi, \prec)$ is strongly metrizable, there exists a weighted graph $H = (U, E_H, w_H)$ such that every path $\pi \in \Pi$ is the unique shortest path between its endpoints in $H$. We may assume w.l.o.g. that the weight of each $\pi \in \Pi$ is smaller than the second-shortest path by an additive factor of $|U| + 1$; otherwise, we can simply scale the weights by an appropriate factor.

We construct the graph $G = (V, E)$ for APSD as follows:
\begin{itemize}
\item For every $u \in U$, construct two vertices $u^{\myin}$ and $u^{\myout}$ in $V$, and add an edge $(u^{\myin}, u^{\myout})$ to $E$.
\item For every edge $(u, v) \in E_H$ where $u \prec v$, add an edge $(u^{\myout}, v^{\myin})$ in $E$.
\end{itemize}
Given an input $\bz \in \{0, 1\}^n$ to the $\bA^{\Pi}$-linear query, we construct the weights $\bw$ for $G$ as follows:
\begin{itemize}
\item For every $u \in U$, let $w((u^{\myin}, u^{\myout})) = z_u$.
\item For every $(u^{\myout}, v^{\myin}) \in E$, let $w((u^{\myout}, v^{\myin})) = w_H(u, v)$. (Note that this is independent of $\bz$.)
\end{itemize}
It is simple to verify that two neighboring inputs $\bz, \bz'$ give rises to neighboring weights $\bw, \bw'$. Furthermore, it is also easy to see that, for every path $\pi = (u = p_1, p_2, \dots, p_\ell = v) \in \Pi$, the unique shortest path between $u^{\myin}$ and $v^{\myout}$ in $G$ is exactly the path $(u^{\myin}, u^{\myout}, p_2^{\myin}, p_2^{\myout}, \dots, v^{\myin}, v^{\myout})$. This means that
\begin{align*}
\dist_G(u_{\myin}, v_{\myin}) = w_H(\pi) + \sum_{p \in \pi} z_p = w_H(\pi) + (\bA^{\Pi}\bz)_{\pi}.
\end{align*}
Therefore, we may run the $(\eps, \delta)$-DP algorithm for APSD to get an estimate $\tdist$ and then output $\tdist(u_{\myin}, v_{\myin}) - w_H(\pi)$ for all $\pi \in \Pi$. This is an $(\eps, \delta)$-DP algorithm for $\bA^{\Pi}$-linear query. Furthermore, if the APSD algorithm is $\alpha$-accurate with probability $1 - \beta$, then the linear query algorithm is also $\alpha$-accurate with probability $1 - \beta$.
\end{proof}

We remark that, if we were working in the directed graph case, then we may start with the graph $H = (U, E_H, w_H)$ that is directed and there would be no need for the ordering $\prec$. However, the ordering is needed for undirected graphs in the proof above to avoid the situation where some edge $(p_i^{\myin}, p_i^{\myout})$ could be skipped in some shortest path in $G$.


\subsection{Lower Bound via Discrepancy of Point-Line System}

An \emph{$(N, D)$-planar point-line system ($(N, D)$-PPLS)} is given by $(\cP, \cL)$ where $\cP = \{p_1, \dots, p_N\} \subseteq \R^2$ is a set of points in the plane and $\cL = \{L_1, \dots, L_D\}$ is a set of lines; when $D = N$, we abbreviate it as $N$-PPLS. The adjacency matrix $\bA^{\cP, \cL} \in \{0, 1\}^{D \times N}$ is such that $\bA^{\cP, \cL}_{L, p}$ is one iff $p \in L$.  The discrepancy of $(\cP, \cL)$ is defined as the discrepancy of $\bA^{\cP, \cL}$.

Chazelle and Lvov~\cite{ChazelleL01a} proved the following lower bound for the point-line system.

\begin{theorem}[\cite{ChazelleL01a}] \label{thm:point-line}
There is an $N$-PPLS with discrepancy $\Omega(N^{1/6})$.
\end{theorem}

Using the same technique as in~\cite{MuthukrishnanN12} it is possible to also show a lower bound for the discrepancy variant where $\gamma \ne 1$ (Definition~\ref{def:discrepancy}).

\begin{corollary} \label{cor:point-line-disc}
For any $\gamma > 0$, there is a PPLS
$(\cP, \cL)$, $|\cP| \leq N$ such that $\disc_\gamma(\bA^{\cP, \cL}) \geq \Omega_\gamma(N^{1/6})$.
\end{corollary}

Before we continue with our lower bound proof, let us note that Chazelle and Lvov~\cite{ChazelleL01a} also gave an upper bound of $O(N^{1/6} \log^{2/3} N)$ for the discrepancy of $N$-PPLS. This leaves a gap of $O(\log^{2/3} N)$ compared to the lower bound. Thanks to the recent advancements in discrepancy theory, we observe that a tight upper bound of $O(N^{1/6})$ can be shown---thereby closing this gap. We provide a proof in \Cref{sec:discrepancy-upper-bound}.

We are now ready to prove our lower bound.

\begin{proof}[Proof of \Cref{thm:lb-main}]
Let $N = \lfloor n/2 \rfloor$, and let $(\cP, \cL)$ be the PPLS as in \Cref{cor:point-line-disc} with $|\cP| \leq N$. Let $U = \cP$ and let $\prec$ denote the lexicographic ordering of the points, and for every $L \in \cL$, add a single path $\pi_L$ to $\Pi$ corresponding to an enumeration of points on $L$ in the order defined by $\prec$. To see that $(U, \Pi, \prec)$ is strongly metrizable, consider the graph $H = (U, E_H, w_H)$ where there is an edge between all points $p, p'$ such that there is no other point on the line segment $pp'$ and the weight of such an edge is $w_H(p, p') = \|p - p'\|_2$. In the graph $H$, the path $\pi_L$ for each line $L$ with starting point $u$ and end point $v$ is the unique shortest path of length $\|u - v\|_2$ between $u$ and $v$. 

Therefore, by \Cref{lem:reduction} and \Cref{thm:linearq-to-disc}, for any $\beta \in (0, 1)$ and any sufficiently small $\gamma, \eps, \delta > 0$, no $(\eps, \delta)$-DP algorithm for APSD can be $\alpha$-accurate w.p. at least $1 - \beta$ for
\begin{align*}
\alpha = \Omega\left(\disc_\gamma(\bA^{\Pi})\right) = \Omega\left(\disc_\gamma(\bA^{\cP, \cL})\right) \geq \Omega\left(n^{1/6}\right). &\qedhere
\end{align*}
\end{proof}

\section{Algorithm with Both Additive and Multiplicative Error}

In this section, we show how to reduce the additive error when we allow multiplicative errors.

\subsection{Private Distance Oracles}

We start by considering the $S$-pair shortest path distances problem and give an improved additive error compared to the algorithms in \Cref{sec:output-pert}. As stated in the Introduction, our algorithm uses an approximate distance oracle due to Thorup and Zwick~\cite{ThorupZ05}, except that we use the  private selection algorithm (\Cref{thm:em}) in each step. We formalize this below for the $\eps$-DP case.

\begin{theorem} \label{thm:pure-dp-do}
Let $k \in \N$ and $\eps > 0$. There is an $\eps$-DP algorithm for the $S$-pair shortest path distances problem that is $\left(2k - 1, O\left(\frac{k^2 |S|^{1 + 1/k} \log^2 n}{\eps}\right)\right)$-accurate w.h.p. 
\end{theorem}

\begin{proof}
For a subset $T \subseteq V$, we use $\dist_G(v, T)$ to denote $\min_{u \in T} \{ \dist_G(v, u) \}$. 

Let $s = |S|$.  Let $q = s^{-1/k}, r = \lceil 10 s^{1/k} \log n \rceil$, and $\eps' = \frac{\eps}{s r k}$.
The algorithm works as follows:
\begin{enumerate}
\item Create sets $S = A_0 \subseteq \cdots \subseteq A_{k - 1}$, where for each $i \in [k - 1]$, each element of $A_{i - 1}$ is included in $A_i$ independently with probability $q$.
\item For each $v \in S$:
\begin{enumerate}
\item Let $B(v) = \emptyset$.
\item For each $i \in \{0\} \cup [k - 1]$:
\begin{enumerate}
\item For each $j \in [r]$:
\begin{enumerate}
\item Use $\eps'$-DP algorithm for selection where the candidate set is $\cC = A_i \setminus B(v)$ and the scoring function is $f_u := \dist(u, v)$. Let $(\hu, \hdist(u, v))$ denote the output. We then add $\hu$ to $B(v)$.
\end{enumerate}
\item Let $p_i(v) := \argmin_{u \in A_i \cap B(v)}\{\hdist(u, v)\}$, where ties are broken arbitrarily.
\end{enumerate}
\end{enumerate}
\item For each $u, v \in V$, the approximate distance $\tdist(u, v)$ is calculated as follows:
\begin{enumerate}
\item Let $u_0 = w_0 = u$ and $v_0 = v$.
\item For each $i \in [k-1]$:
\begin{enumerate}

\item If $w_{i-1} \in B(v_{i-1})$, break.
\item Otherwise, let $u_i = v_{i - 1}, v_i = u_{i - 1}$, and $w_i = p_i(u_i)$.
\end{enumerate}
\item Let $\tdist(u, v) := \hdist(w, u) + \hdist(w, v)$, where $w$ denotes the last $w_i$ that got assigned in the above.
\end{enumerate}
\end{enumerate}
The basic composition theorem implies that our algorithm is $\eps$-DP as desired. We will next analyze its utility. Let $\rnn_G(v, A_i)$ denote the distance between $v$ to the $r$th closest vertex in $A_i$. The utility guarantees from \Cref{thm:em} implies that, with probability $1 - O(1/n^3)$, we have
\begin{align} \label{eq:bottom-acc}
\forall u \in A_i \text{ such that } \dist_G(v, u) \leq \rnn_G(v, A_i) - \alpha, & &u \in B(v),
\end{align}
and
\begin{align} \label{eq:apx-distance-bottom}
\forall u \in A_i,& &|\dist_G(u, v) - \hdist(u, v)| \leq \alpha,
\end{align}
where $\alpha = O\left(\frac{\log n}{\eps'}\right) = O\left(\frac{k s^{1+1/k} \log^2 n}{\eps}\right)$.

Furthermore, notice that, for all $i \in \{0\} \cup [k - 2]$, since each element in $A_i$ is kept in $A_{i+1}$ with probability $q$, we have
\begin{align*}
\Pr[\rnn_G(v, A_i) > \dist_G(v, A_{i+1})] = (1 - q)^r < 1/n^3.
\end{align*}
Therefore, by a union bound, with probability $1 - O(1/n)$ we have
\begin{align*}
\forall v \in V, i \in \{0\} \cup [k - 2], & &\rnn_G(v, A_i) \leq \dist_G(v, A_{i+1}),
\end{align*}
which combined with~\eqref{eq:bottom-acc} implies that, for all $i \in \{0\} \cup [k-2]$, we have
\begin{align} \label{eq:close-vertices-included}
\forall u \in A_i \text{ such that } \dist_G(v, u) \leq \dist_G(v, A_{i+1}) - \alpha, & &u \in B(v).
\end{align}

Similarly, notice that each element is included in $A_{k-1}$ w.p. $q^{k-1} = s^{1/k}/s$ independently. Therefore, with probability $1 - O(1/n)$ we have $|A_{k-1}| \leq r$, which means that
\begin{align*}
A_{k-1} \subseteq B(v).
\end{align*}
Note that this event means that the distance calculation is valid because if we reach the end of the loop (i.e., reaching $i = k - 1$), we will always have $w \in A_{k - 1} \subseteq B(u), B(v)$.

Henceforth, we will assume that these events occur.

Now, consider any $u, v \in V$. Since $\tdist(u, v) := \hdist(w, u) + \hdist(w, v)$ for some $w \in V$, we may lower bound $\tdist(u, v)$ by
\begin{align*}
\tdist(u, v) \overset{\eqref{eq:apx-distance-bottom}}{\geq} \dist(w, u) + \dist(w, v) -2\alpha \geq \dist(u, v) - 2\alpha.
\end{align*}
As for the upper bound on $\tdist(u, v)$, let $\ell \in [k]$ denote the largest value of $i$ for which $u_i, v_i$ are assigned (before the loop breaks/ends). For any $i \in [\ell]$, since the loop does not break, we have $w_{i-1} \notin B(v_{i-1})$. From \eqref{eq:close-vertices-included}, this implies that
\begin{align*}
\dist_G(w_{i-1}, v_{i-1}) \geq \dist_G(v_{i-1}, A_i) - \alpha.
\end{align*}
Furthermore, the definition of $p_i(v_{i-1})$ together with~\eqref{eq:apx-distance-bottom} implies that
\begin{align*}
\dist_G(w_i, v_{i-1}) = \dist_G(p_i(v_{i-1}), v_{i-1}) \leq \dist_G(v_{i-1}, A_i) + \alpha.
\end{align*}
Therefore, we have
\begin{align}
\dist_G(w_i, u_i) = \dist_G(w_i, v_{i-1}) &\leq \dist_G(w_{i-1}, v_{i-1}) + 2\alpha \nonumber \\
&\leq \dist_G(w_{i-1}, u_{i-1}) + \dist_G(u_{i-1}, v_{i-1}) + 2\alpha \nonumber \\
&=  \dist_G(w_{i-1}, u_{i-1}) + \dist_G(u, v) + 2\alpha, \label{eq:distance-grow-each-level}
\end{align}
where the second inequality follows from the triangle inequality.

As a result, we have
\begin{align*}
\tdist(u, v) &= \hdist(u, w_\ell) + \hdist(v, w_\ell) \\
&= \hdist(u_\ell, w_\ell) + \hdist(v_\ell, w_\ell) \\
&\overset{\eqref{eq:apx-distance-bottom}}{\leq}  \dist_G(u_\ell, w_\ell) + \dist_G(v_\ell, w_\ell) + 2\alpha \\
(\text{Triangle inequality}) &\leq 2 \dist_G(u_\ell, w_\ell) + \dist_G(u, v) + 2\alpha \\
&\overset{\eqref{eq:distance-grow-each-level}}{\leq} 2 \dist_G(u_0, w_0) + (2\ell + 1) \cdot \dist_G(u, v) + 2(2\ell + 1)\alpha \\
(\text{Because } w_0 = u_0) &= (2\ell + 1) \cdot \dist_G(u, v) + 2(2\ell + 1)\alpha \\
&\leq (2k - 1) \cdot \dist_G(u, v) + (4k - 2) \alpha.
\end{align*}

Therefore, we can conclude that w.p. $1 - O(1/n)$, the output is $(2k - 1, (4k - 2)\alpha)$-accurate.
\end{proof}

The $(\eps, \delta)$-DP case is nearly identical except we use advanced composition instead of basic composition.

\begin{theorem} \label{thm:apx-dp-do}
Let $k \in \N$ and $\eps, \delta \in [0, 1)$. There is an $\eps$-DP algorithm for the $S$-pair shortest path distances problem that is $\left(2k - 1, O\left(\frac{k^{3/2} \log^{3/2} n \sqrt{|S|^{1+1/k} \log(1/\delta)}}{\eps}\right)\right)$-accurate w.h.p. 
\end{theorem}

\begin{proof}
The proof is exactly the same as in that of \Cref{thm:pure-dp-do} except that here we use an $\eps''$-DP selection algorithm with $\eps'' := \frac{\eps}{2\sqrt{2(srk)\log(2/\delta)}}$.  Advanced composition  (\Cref{thm:advanced-composition}) implies that the algorithm is $(\eps, \delta)$-DP. As for the error, we now have $\alpha = O\left(\frac{\log n}{\eps''}\right) = O\left(\frac{\sqrt{k s^{1+1/k} \log(1/\delta)} \log^{3/2} n}{\eps}\right)$, resulting in the claimed additive error.
\end{proof}

\subsection{Combining with Input Perturbation Algorithm}

By plugging~\Cref{thm:pure-dp-do} into \Cref{lem:combining} with $s = (n/k^2)^{k/(2k + 1)}$, we immediately arrive at the following, which is a more detailed version of \Cref{thm:pure-mult-intro}.

\begin{theorem} \label{thm:pure-dp-alg-mult}
Let $k \in \N$. For any $\eps > 0$, there exists an $\eps$-DP algorithm for APSD that is $\left(2k - 1, O\left(\frac{k \cdot n^{(k+1)/(2k+1)} \log^2 n}{\eps}\right)\right)$-accurate w.h.p.
\end{theorem}

Similarly, by plugging~\Cref{thm:apx-dp-do} into \Cref{lem:combining} with $s = \left(\frac{n \log^{1/2} n}{k^{3/2} \log^{1/2}(1/\delta)}\right)^{\frac{2k}{3k+1}}$, we immediately arrive at the following, which is a more detailed version of \Cref{thm:apx-mult-intro}.

\begin{theorem} \label{thm:apx-dp-alg-mult}
Let $k \in \N$. For any $\eps, \delta \in (0, 1)$, there exists an $(\eps, \delta)$-DP algorithm for APSD that is $\left(2k - 1, O\left(\frac{k \cdot n^{(k+1)/(3k+1)} (\log n)^{(5k+2)/(3k+1)} (\log(1/\delta))^{k/(3k+1)}}{\eps}\right)\right)$-accurate w.h.p.
\end{theorem}

\section{Discussion and Open Questions}

In this work, we give an algorithm for the APSD problem with error $\tilde{O}(n^{2/3}/\eps)$ in the $\eps$-DP setting and $\tilde{O}(\sqrt{n} / \eps)$ in the $(\eps, \delta)$-DP setting, together with a lower bound of $\Omega(n^{1/6})$ for any sufficiently small $\eps, \delta > 0$. 

An obvious open problem here is to close the gap between the lower and upper bounds. On this front, we note that our lower bound technique cannot go beyond an $\tilde{O}(n^{1/4})$ additive error. The reason is that, for any strongly metrizable path system $(U, \Pi)$, $\bA^{\Pi}$ has VC dimension $d \leq 2$, since unique shortest paths cannot ``fork and rejoin''. Thus, we can immediately use an algorithm of~\cite{MuthukrishnanN12}\footnote{In fact,~\cite{MuthukrishnanN12} gives a root mean square error guarantee. However, one may use the boosting framework of~\cite{DworkRV10} to achieve a similar (up to polylogarithmic factors) bound for the maximum error.} to get an error of $\tilde{O}_{\eps, \delta}\left(n^{1/2 - 1/(2d)}\right) \leq \tilde{O}_{\eps, \delta}(n^{1/4})$. It remains interesting whether a lower bound of $\tilde{\Omega}_{\eps, \delta}(n^{1/4})$ can be shown via our technique. A possibly helpful tool in this direction is the work of~\cite{Bodwin19}, which characterizes the set of strongly metrizable path systems based on a family of (infinite) prohibited subgraphs; it may be possible to use this characterization to construct a more elaborate graph with higher discrepancy than the ones we obtain from PPLS.

When a multiplicative approximation is allowed, we give algorithms with improved, but still polynomial in $n$, additive errors.
As mentioned earlier, our lower bound does \emph{not} apply in this case, and it remains a possibility that an algorithm with, e.g., $O(1)$ multiplicative approximation and polylog($n)$  additive error exists. Coming up with such an algorithm or refuting its existence is also an interesting future direction.

\bibliographystyle{alpha}
\bibliography{ref}

\appendix

\section{Discrepancy Bounds for Planar Point-Line Systems}

\subsection{Tight Upper Bound}
\label{sec:discrepancy-upper-bound}

As mentioned earlier, Chazelle and Lvov~\cite{ChazelleL01a} prove an upper bound of $O(N^{1/6} \log^{2/3} N)$ on the discrepancy of any $N$-PPLS. Below, we improve this upper bound to $O(N^{1/6})$, which is tight in light of the lower bound in~\cite{ChazelleL01a}.

\begin{theorem} \label{thm:tight-discrepancy-point-line}
For any $(N, D)$-PPLS $(\cP, \cL)$, we have $\disc(\bA^{\cP, \cL}) \leq O(N^{1/6})$.
\end{theorem}

The proof of \Cref{thm:tight-discrepancy-point-line} is similar to that of~\cite{ChazelleL01a}, which uses the Szemerédi--Trotter theorem~\cite{SzemerediT83} (stated below) that bounds the number of lines that pass through many points; the main difference is that we will use the partial coloring result of~\cite{LovettM15} (also stated below) which is stronger than that in the original proof of~\cite{ChazelleL01a}.

\begin{theorem}[Szemerédi--Trotter~\cite{SzemerediT83}] \label{thm:st}
There exists a constant $c \geq 1$ such that, for any $N, k \in \N$ with $k \geq 2$ and $N$ points $x_1, \dots, x_N \in \R^2$, the number of lines that pass through at least $k$ points is at most $c(N^2/k^3 + N/k)$.
\end{theorem}

\begin{theorem}[Partial Coloring~\cite{LovettM15}] \label{thm:partial-coloring}
Let $\bv_1, \dots, \bv_m \in \R^n$ be vectors, $\bx^{\mystart} \in [-1, +1]^n$ be the ``starting'' vector and $c_1, \dots, c_m > 0$ be thresholds\footnote{Note that the original formulation of~\cite{LovettM15} also has an ``approximation error'' parameter $\delta$ and the guarantee is only that $|x^{\myend}_i| \geq 1 - \delta$ for more than half of the coordinates. However, we may derive our formulation by simply halving each of $c_j$ and picking $\delta = \max\{c_1, \dots, c_m\} / n$. The reason their work contains $\delta$ is mainly due to the algorithmic aspect, i.e., the running time grows as $(1/\delta)^{O(1)}$.  Since we are not interested in computational complexity, we choose to discard $\delta$ from the formulation.}. If $\sum_{j \in [m]} \exp(-c_j^2 / 16) \leq n/16$, then there exists a point $\bx^{\myend} \in [-1, +1]^n$ such that
\begin{enumerate}
\item[(i).] $|\left<\bx^{\myend} - \bx^{\mystart}, \bv_j\right>| \leq c_j \|\bv_j\|_2$ for all $j \in [m]$,
\item[(ii).] $x^{\myend}_i \in \{\pm 1\}$ for at least $n/2$ indices $i \in [n]$.
\end{enumerate}
\end{theorem}

Compared to the partial coloring result used in~\cite{ChazelleL01a} (proved in~\cite{beck_chen_1987,Chazelle-book}), \Cref{thm:partial-coloring} is sharper in that it allows for $O(n)$ ``stringent'' constraints where $c_j = o(1)$, whereas the version used in~\cite{ChazelleL01a} only allows $O(n / \log n)$ such constraints. When it comes to Szemerédi--Trotter theorem, this means that \cite{ChazelleL01a} can only let the lines with at least $k \geq \Theta(N^{1/3} \log^{1/3} N)$ points be these ``stringent'' constraints; on the other hand, \Cref{thm:partial-coloring} allows us to take $k = \Theta(N^{1/3})$. This is indeed the reason for the improvement achieved here.

\begin{proof}[Proof of \Cref{thm:tight-discrepancy-point-line}]
We construct the final vector $\bx$ via the standard partial coloring approach as follows:
\begin{itemize}
\item Let $t = 0$ and $\bx^0 = (0, \dots, 0)$ be the all-zeros vector in $\R^N$.
\item While $\bx^t \notin \{\pm 1\}^N$:
\begin{itemize}
\item Let $I_t = \{i \in [N] \mid x^t_i \notin \{\pm 1\}\}$ denote the set of fractional coordinates, let $\cP_t = \{p_i \mid i \in I_t\}$ be the set of corresponding points, and let $n_t = |I_t|$.
\item Let $\cL_t$ denote the set of lines that pass through at least two points in $\cP_t$. For each $L_j \in \cL^t$, let $\bv_j^t$ denote the $j$th row in $\bA^{\cP_t, \cL}_{I_t}$. Furthermore, let $c^t_j = \frac{12800 \cdot c n_t^{1/6}}{ \|\bv_j^t\|_2}$, where $c$ is the constant from \Cref{thm:st}.
\item Apply \Cref{thm:partial-coloring} on $\{\bv_i^t\}_{L_j \in \cL^t}$ with the starting vector $\bx^{\mystart} = \bx^t_{I_t}$ and thresholds $\{c^t_j\}_{L_j \in \cL^t}$ to get $\bx^{\myend} \in [-1, +1]^{I_t}$.
\item Finally, let $\bx^{t + 1}$ be such that
\begin{align*}
x^{t + 1}_j &=
\begin{cases}
x^t_j &\text{ if } j \notin I_t, \\
x^{\myend}_j &\text{ otherwise.}
\end{cases}
\end{align*}
\item $t = t + 1$.
\end{itemize}
\item Let $\bx = \bx^t$ be the coloring.
\end{itemize}
Let us start by checking that the invocation of \Cref{thm:partial-coloring} is valid, i.e., that $\sum_{L_j \in \cL^t} \exp(-c^t_j / 16) \leq n^t/16$. Let $q = \lceil 64 c n_t^{1/3} \rceil$. We may bound the sum using the Szemerédi--Trotter theorem as follows:
\begin{align*}
\sum_{L_j \in \cL^t} \exp \left( \frac{-c^t_j}{16} \right) &= \sum_{L_j \in \cL^t} \exp \left( \frac{-800 \cdot c n_t^{1/6}}{\|\bv_j^t\|_2} \right) \\
&= \sum_{k=2}^{n_t} \sum_{L_j \in \cL^t \atop |L_j \cap \cP_t| = k} \exp \left( \frac{-800 \cdot c n_t^{1/6}}{\sqrt{k}} \right) \\
&= \sum_{k = q}^{n_t} \sum_{L_j \in \cL^t \atop |L_j \cap \cP_t| = k} \exp \left( \frac{-800 \cdot c n_t^{1/6}}{\sqrt{k}} \right) + \sum_{k = 2}^{q - 1} \sum_{L_j \in \cL^t \atop |L_j \cap \cP_t| = k} \exp \left( \frac{-800 \cdot c n_t^{1/6}}{\sqrt{k}} \right) \\
&\leq |\left\{L_j \in \cL^t \mid |L_j \cap \cP_t| \geq q\right\}| + \sum_{s=1}^{\lfloor \log(q-1) \rfloor} \sum_{L_j \in \cL^t \atop |L_j \cap \cP_t| \geq 2^s} \exp \left( \frac{-800 \cdot c n_t^{1/6}}{2^{(s+1)/2}} \right) \\
(\text{\Cref{thm:st}}) &\leq c\left( \frac{n_t^2}{q^3} + \frac{n_t}{q} \right) + \sum_{s=1}^{\lfloor \log(q-1) \rfloor} c\left( \frac{n_t^2}{2^{3s}} + \frac{n_t}{2^s} \right) \cdot \exp \left( \frac{-800 \cdot c n_t^{1/6}}{2^{(s+1)/2}} \right) \\
&= \frac{n_t}{32} + \sum_{s=1}^{\lfloor \log(q-1) \rfloor} 4097\cdot c \frac{n_t^2}{2^{3s}} \cdot \exp\left(-560 \cdot c \frac{n_t^{1/6}}{2^{s/2}} \right). \\
\end{align*}
Now, let $a(s) = 70cn_t^{1/6}/2^{s/2}$. For $s \leq \log(q-1)$, we have $a(s) \geq 70 cn_t^{1/6}/\sqrt{q-1} \geq 1$. From this, we have $e^{a(s)} \geq a(s)$. In other words, $e^{-a(s)} \leq 1/a(s)$. Plugging this back into the above, we have
 \begin{align*}
\sum_{L_j \in \cL^t} \exp \left(\frac{-c^t_j}{16} \right) 
&\leq \frac{n_t}{32} + \sum_{s=1}^{\lfloor \log(q-1) \rfloor} 4097 \cdot c \frac{n_t^2}{2^{3s}} \cdot \left(\frac{1}{70cn_t^{1/6}/2^{s/2}}\right)^8. \\
&\leq \frac{n_t}{32} + \frac{n_t^{2/3}}{10^6 c} \cdot \sum_{s=1}^{\lfloor \log(q-1) \rfloor} 2^s \\
&\leq \frac{n_t}{32} + \frac{n_t^{2/3}}{10^6 c} \cdot 2^{\lfloor \log(q-1) \rfloor + 1} \\
&\leq \frac{n_t}{32} + \frac{n_t^{2/3}}{10^6 c} \cdot 128 \cdot cn_t^{1/3} \\
&\leq \frac{n}{16},
\end{align*}
as desired.

As for the discrepancy, we have
\begin{align*}
\|\bA^{\cP, \cL} \bx\|_\infty 
&= \max_{L_j \in \cL} \left<\bv^0_j, \bx\right> \\
&\leq \max_{L_j \in \cL} \left(1 + \sum_{t \geq 0} \left<\bv^t_j, \bx^{t + 1} - \bx^t\right>\right)\\
(\text{\Cref{thm:partial-coloring}(i)}) &\leq \max_{L_j \in \cL} \left(1 + \sum_{t \geq 0} O \left(n_t^{1/6} \right)\right) \\
(\text{\Cref{thm:partial-coloring}(ii)}) &\leq \max_{L_j \in \cL} \left(1 + \sum_{t \geq 0} O\left(\left(\frac{N}{2^t}\right)^{1/6}\right)\right) \\
&\leq O(N^{1/6}).
\end{align*}
Note that the extra $1$ on the second line comes from the fact that we only include in $\cL^t$ lines that pass through at least two points in $\cP_t$.
\end{proof}

\subsection{Lower Bound for $\gamma < 1$}

In this section we prove a lower bound for discrepancy of partial coloring, i.e., for $\gamma < 1$.

\begin{proof}[Proof of \Cref{cor:point-line-disc}]
Let $c_1 > 0$ be the constant from~\Cref{thm:point-line} such that for all $N \in \N$, there is an $N$-PPLS $(\cP, \cL)$ such that $\disc(\bA^{\cP, \cL}) \geq c_1 N^{1/6}$. Furthermore, let $c_2 > 0$ be the constant from~\Cref{thm:tight-discrepancy-point-line} such that for all $N \in \N$, for 
any $N$-PPLS $(\cP', \cL')$ we have $\disc(\bA^{\cP', \cL'}) \leq c_2 N^{1/6}$.

Let $c = \frac{c_1}{12} \cdot \frac{\log(1/(1-\gamma))}{\log(2c_2/c_1)}$. We will show that for all $N \in \N$, there exists a PPLS 
$(\cP^*, \cL^*)$ with $|\cP^*| \leq N$ such that $\disc_\gamma(\bA^{\cP^*, \cL^*}) \geq c N^{1/6}$. Suppose for the sake of contradiction that there exists $N \in \N$ such that for all PPLS $(\cP^*, \cL^*)$ with $|\cP^*| \leq N$, we have $\disc_\gamma(\bA^{\cP^*, \cL^*}) < cN^{1/6}$.

Let $(\cP, \cL)$ be the $N$-PPLS guaranteed by \Cref{thm:point-line}.
Let $\hN = N(0.5c_1 / c_2)^6$.
Consider the following recursive coloring algorithm for $(\cP, \cL)$:  
\begin{itemize}
\item Let $t = 0$ and $\cP_0 = \cP$.
\item While $|\cP_t| > \hN$:
\begin{itemize}
\item Apply the hypothesis on $(\cP_t, \cL)$ to get a partial coloring $\bx^t \in \{-1, 0, +1\}^{\cP_t}$ such that $\|\bx^t\|_1 \geq \gamma |\cP_t|$ and $\|\bA^{\cP_t, \cL}\bx^t\|_\infty \leq c N^{1/6}$.
\item Let $\cP_{t + 1} := \{p \in \cP_t \mid x^t_p = 0\}$.
\item $t = t + 1$.
\end{itemize}
\item Apply~\Cref{thm:tight-discrepancy-point-line} on $(\cP_t, \cL)$ to get a coloring $\bx^t \in \{\pm 1\}^{\cP_t}$ such that $\|\bA^{\cP_t, \cL}\bx^t\|_\infty \leq c_2 \hN^{1/6}$.
\item Let $\bx$ be the coloring generated above,  i.e., $x_p = x^{t'}_p$ for all $p \in \cP_{t'} \setminus \cP_{t'+1}$.
\end{itemize}
Let $T$ denote the number of iterations executed by the algorithm.  Since $|\cP_t|$ reduces by a factor of (at least) $1 / (1 - \gamma)$ each time, we have $T \leq \frac{\log(N/\hN)}{\log(1/(1-\gamma))} \leq \frac{6 \log(2c_2/c_1)}{\log(1/(1-\gamma))}$. Therefore, we have
\begin{align*}
\|\bA^{\cP, \cL}\bx\|_\infty < T \cdot cN^{1/6} + c_2\hN^{1/6}
\leq 0.5 c_1 N^{1/6} + 0.5c_1 N^{1/6} = c_1N^{1/6},
\end{align*}
which contradicts our choice of $(\cP, \cL)$.
\end{proof}

\end{document}